\documentclass[a4paper, UKenglish]{lipics-v2018}

\usepackage{verbatim}

\usepackage{amsfonts}
\usepackage{braket}
\usepackage{amsmath}
\usepackage{amsthm}
\usepackage{amssymb}
\usepackage{mathabx}
\newtheorem{conj}{\bf Conjecture}
\newtheorem{fact}[theorem]{\bf Fact}
\newtheorem{claim}{\bf Claim}
\newtheorem*{claim*}{Claim}
\usepackage{mathtools}

\DeclarePairedDelimiter\floor{\lfloor}{\rfloor}

\usepackage{microtype}

\title{The Power of One Clean Qubit in Communication Complexity}

\titlerunning{Communication Complexity with One Clean Qubit}

\author{Hartmut Klauck}{Centre for Quantum Technologies, Singapore}{hklauck@gmail.com}{}{This work is funded by the Singapore Ministry of Education and by the Singapore National Research Foundation. Also supported by Majulab UMI 3654.}

\author{Debbie Lim}{Centre for Quantum Technologies, Singapore}{limhueychih@gmail.com}{}{}

\authorrunning{H. Klauck and D. Lim}

\Copyright{H. Klauck and D. Lim}



\category{}

\relatedversion{}

\supplement{}

\funding{}

\acknowledgements{}

\EventEditors{}
\EventNoEds{2}
\EventLongTitle{}
\EventShortTitle{}
\EventAcronym{CVIT}
\EventYear{}
\EventDate{}
\EventLocation{}
\EventLogo{}
\SeriesVolume{43}
\ArticleNo{1}
\nolinenumbers 
\hideLIPIcs  

\begin{document}

\maketitle

\begin{abstract}
We study quantum communication protocols, in which the players' storage starts out in a state where one qubit is in a pure state, and all other qubits are totally mixed (i.e. in a random state), and no other storage is available (for messages or internal computations). This restriction on the available quantum memory has been studied extensively in the model of quantum circuits, and it is known that classically simulating quantum circuits operating on such memory is hard when the additive error of the simulation is exponentially small (in the input length), under the assumption that the polynomial hierarchy does not collapse.

We study this setting in communication complexity. The goal is to consider larger additive error for simulation-hardness results, and to not use unproven assumptions.

We define a complexity measure for this model that takes into account that standard error reduction techniques do not work here. We define a clocked and a semi-unclocked model, and describe efficient simulations between those.

We characterize a one-way communication version of the model in terms of weakly unbounded error communication complexity.

Our main result is that there is a quantum protocol using one clean qubit only and using $O(\log n)$ qubits of communication, such that any classical protocol simulating the acceptance behaviour of the quantum protocol within additive error $1/poly(n)$ needs communication $\Omega(n)$.

We also describe a candidate problem, for which an exponential gap between the one-clean-qubit communication complexity and the randomized complexity is likely to hold, and hence a classical simulation of the one-clean-qubit model within {\em constant} additive error might be hard in communication complexity. We describe a geometrical conjecture that implies the lower bound.

 \end{abstract}

\section{Introduction}

The computational power of quantum models of computation with different memory restrictions has been studied in order to understand the use of imperfectly implemented qubits. Some possible types of memory restrictions include having only few qubits that are in a pure state plus an abundance of qubits that start in the totally mixed state \cite{Hard}, having memory that starts in an incompressible state that needs to be returned unchanged at the end of the computation, plus some limited auxiliary space available \cite{Buhrman}, or simply having very little memory for the computation \cite{ambainis,ksw:dpt-siam,klauck:tradeoffs,asw:adversarydpt}.
The underlying idea in these topics is to study the power of models of quantum computing in which the quantum memory is weak, but the control of this memory is good. This is in contrast to the study of models of quantum computation, where the underlying memory is good, but the control is weak, or restricted, such as the Boson-Sampling model \cite{Boson}. Both are a step towards understanding the power of quantum computing models that are closer to being implementable than the standard circuit model, and eventually to demonstrate quantum supremacy (i.e., to show that for some problem (of possibly small practical interest) quantum computers that can be built outperform classical computers demonstrably).

This paper explores the potential of a model of quantum communication that uses memory containing only a small number of qubits that start in a known pure state, in particular the power of a having only a single clean qubit (plus many qubits that start in the totally mixed state, i.e., start in a random state).

The one-clean-qubit model originally proposed by Knill and Laflamme \cite{KnillandLaflamme} is a model of quantum computing where the memory starts in the tensor product of a single qubit in a pure state $\ket{0}$ with the other $m$ qubits that are in the completely-mixed state, with no further storage allowed. This initial state is described by the density matrix

\[\rho=\ket{0}\bra{0}\otimes\frac{I}{2^m}.\]

The model was originally motivated by the nuclear magnetic resonance (NMR) approach to quantum computing, where the initial state may be highly mixed.
Quantum circuits operating on such memory  are able to perform tasks that look hard classically, such as estimating Jones polynomials, computing Schatten $p$-norms, spectral density approximation, testing integrability, computation of fidelity decay \cite{KnillandLaflamme, Shor, pnorms, decay, integrability}, just to name a few. Recently, K.~Fujii et al.~showed that quantum circuits under the one-clean-qubit restriction cannot be efficiently classically simulated unless the polynomial hierarchy collapses to the second level \cite{Hard}. In other words, assuming that the polynomial hierarchy does
 not collapse, polynomial size quantum circuit operating under the one-clean qubits restriction can have acceptance/rejection probabilities such that any classical randomized circuit that has the same acceptance/rejection probabilities up to additive error $1/exp(n)$ must have superpolynomial size.
 We note here that we will not consider simulations with multiplicative error in this paper, since those pose a much stronger requirement on the simulation, for instance the simulating algorithm must replicate events of tiny probability with approximately the same probability, and hence such simulations are much less interesting.

 In this paper, we study the hardness of simulating the one-clean-qubit model classically in the model of  communication complexity. We will consider simulations of the one-clean-qubit model with different amounts of additive errors, namely $\frac{1}{poly(n)}$ and $\Omega(1)$.


\subsection{Organization}
After some preliminaries in Section \ref{2},
in Section \ref{3} we discuss related work. In Section \ref{results}, we sketch our results. In Section \ref{r1}
we develop our model of quantum communication with one clean qubit. We motivate the main complexity measure and
introduce the concepts of clocked and semi-unclocked protocols. Section \ref{ow} is about our characterization of one-way communication complexity in our model. Section \ref{r2} discusses our main result, which concerns the hardness of classically simulating the one-clean-qubit model with additive error.

\section{Preliminaries}\label{2}
\subsection{Communication Complexity}
Yao's \cite{Yao} model of communication complexity consists of two players, Alice and Bob, who are each given private inputs $x\in X$ and $y\in Y$ respectively. In addition, they both know the function $f$ and agree to a certain communication protocol beforehand. The task they wish to perform is to compute $z=f(x, y)$. Having no knowledge of each others' inputs, they have to communicate with each other in order to obtain the result $z$. Communication complexity asks the question  "how much communication is needed to compute $f(x, y)$?", and assumes that the players have unlimited computational power.

For formal definitions regarding standard types of communication protocols see \cite{kushilevitz&nisan:cc}, regarding quantum communication complexity see \cite{ViS}.
We will use the following notations:

\begin{definition}
$Q(f), R(f)$ denote the quantum (without entanglement) and randomized (with public coin) communication complexities of a function $f$ with error $1/3$. A subscript like $Q_\epsilon(f)$ denotes other errors $\epsilon$.
\end{definition}

\section{Related Work}\label{3}
There has been a lot of research focusing on the hardness of classical simulations of restricted models of quantum computing under certain assumptions \cite{IQP, Boson, arthumerlin, cqc, cliffordcircuits, morimae, toffoli, bs, unlikely}. That is to say, a reasonable assumption in complexity theory leads to the impossibility of efficient sampling by a classical computer according to an output probability distribution that can be generated by a quantum computation model. For instance, it is proven that classical simulation with multiplicative error of the IQP model \cite{IQP} and Boson sampling \cite{Boson} is hard, unless the polynomial-time hierarchy collapses.

It is interesting to ask if such a result holds for the one-clean-qubit model as well. Over the past few years, the one-clean-qubit model has be shown to be capable of efficiently solving problems where no efficient classical algorithm is known, such as estimating Jones polynomials, computing Schatten $p$-norms, spectral density approximation, testing integrability and computation of fidelity decay \cite{KnillandLaflamme, Shor, pnorms, decay, integrability}. It has been conjectured that the one-clean-qubit model can be more powerful than classical computing for some problems. However, there has been no proof for such a conjecture. In \cite{classsim}, T. Morimae and K. Koshiba showed that if the output probability distribution of the one-clean-qubit model can be classically efficiently approximated (with at most an exponentially small additive error) then $BQP\subseteq BPP$. Although the belief that $BQP\neq BPP$ is maybe less strong than that of $P\neq NP$ or that the polynomial hierarchy does not collapse, there is still a good case for it and the assumption is necessary for simulation hardness anyway. Therefore the results in \cite{classsim} suggest that the one-clean-qubit model is unlikely to be classically efficiently simulatable with exponentially small additive error.

T. Morimae et al. introduced $DQC1_k$, a modified version of the one-clean-qubit model where the workspace starts with one clean qubit and $k$ qubits are measured at the end of the computation. They showed that the $DQC1_k$ model cannot be efficiently classically simulated for $k\geq 3$ (within constant {\em multiplicative} error) unless the polynomial hierarchy collapses \cite{morimae}.

Recently, K. Fujii et al. showed via circuit complexity that the one-clean-qubit model cannot be efficiently classically simulated with $\frac{1}{exp(n)}$ additive error unless the polynomial hierarchy collapses to the second level \cite{Hard}. 

All existing results regarding the efficient classical simulation of the one-clean-qubit model are conditional (e.g.~rely on non-collapse of the polynomial hierarchy) and require simulations to have exponentially small additive error.

We also mention work on classical memory-restricted communication complexity (e.g.~\cite{spcc}) in which some similar issues appear as in this work.

\section{Overview of Results}\label{results}
\begin{itemize}
    \item Definition of a complexity measure for the one-clean-qubit model in communication complexity:\\
    The complexity measure (cost) of a one-clean-qubit protocol is given by $c\cdot\big(\frac{1}{\epsilon^2}\big)$, where $c$ is the communication and $\epsilon$ is the bias. We define a clocked and a semi-unclocked version.
    \item Simulation of a clocked $k$-clean-qubit models using only one-clean qubit is inexpensive:\\
    Such simulations cost  $(c+1)\cdot\big(\frac{2^k}{\epsilon}\big)^2$, where $c$ is the communication and $\epsilon$ is the bias.
    \item The clocked $k$-clean-qubit model can be simulated by the semi-unclocked one-clean-qubit model:\\
    Such simulations incur a cost of  $O(c\log{c})\cdot\big(\frac{2^k}{\epsilon}\big)^2$, where $c$ is the communication and $\epsilon$ is the bias.
    \item Upper and lower bounds on the complexity measure of the one-way one-clean-qubit communication complexity model:\\
    The complexity measure of the one-way one-clean-qubit communication complexity model denoted as $Q_{[1]}^{A\to B}(f)$ is bounded by $2^{\Omega(PP(f))-O(\log{n})}\leq Q_{[1]}^{A\to B}(f)\leq 2^{O(PP(f))}$.
    \item Classically simulating the one-clean-qubit model with $\frac{1}{poly(n)}$ additive error requires an exponential increase in communication:\\
    We consider the $MIDDLE$ problem and give a quantum protocol with one-clean qubit that requires $O(\log{n})$ communication while any classical simulation with $\frac{1}{poly(n)}$ additive error requires $\Omega(n)$ communication.

We stress that in previous results about the hardness of simulating the one-clean-qubit model (in circuit complexity)
the additive error must be of size at most $1/exp(n)$ for the simulation to be hard, which stems from low probability events being considered that one would never observe realistically. That means that running the one-clean-qubit circuit as an experiment, and observing an outcome that contradicts classicality is an event that happens only with exponentially small probability, and the classical simulation is only hard because of such extremely low probability events. Our result also uses low probability events, but $1/poly(n)$ is much more reasonable, and the events are observable when repeating such a protocol $poly(n)$ times.

    \item Simulating the one-clean-qubit model with constant additive error:\\
    We consider a problem $ABC$ as a candidate to show that simulating the one-clean-qubit model with {\em constant} additive error is hard, and construct a quantum protocol that requires $O(\log{n})$ communication using one clean qubit for $ABC$. We conjecture that any classical simulation with constant additive error requires $\Omega(\sqrt{n})$ communication and give a matching upper bound.
\end{itemize}
\textbf{Disclaimer: All $I$'s used in this paper are identity matrices whose dimensions are clear from the context.}

\section{Communication Complexity of the One-Clean-Qubit Model}\label{r1}
\subsection{The One-Clean-Qubit Model}
\begin{definition}[$k$-Clean-Qubit Model]
   In a $k$-clean-qubit protocol, all storage initially consists  of only $k$ qubits in a clean state $\ket{0}$, while the rest ($m$ qubits) are in the totally mixed state. The players communicate as in a standard quantum protocol. Only at the end of the computation, a single, arbitrary projective measurement (not depending on the inputs) is performed.
\end{definition}

 By this definition, all storage in the one-clean-qubit model consists of only one qubit in a clean state $\ket{0}$, while the rest ($m$ qubits) are in the totally mixed state. This can  be described by the density matrix
\begin{equation}\label{model}
    \rho=\ket{0}\bra{0}\otimes\frac{I}{2^m}.
\end{equation}

 A protocol in this model for a function $f$ communicates $c$ qubits. Assume the protocol has a bias of $\epsilon$ and hence an error of $\frac{1}{2}-\epsilon$. In general, it is not possible to improve the error to, say, $\frac{1}{3}$. Following \cite{Shor}, we therefore allow the computation to be repeated (virtually) $O(\frac{1}{{\epsilon}^2})$ times until a correctness probability of at least $\frac{2}{3}$ is achieved, and therefore define the cost of the (unrepeated) protocol to be $c\cdot(\frac{1}{\epsilon})^2$ qubits.

 \begin{definition}[$Q_{[1]}(f)$]\label{compmeasure}
 Let $\cal P$ denote a one-clean-qubit clocked (explained later) protocol for a function $f:X\times Y\to\{0,1\}$, such that 0-inputs are accepted with probability at most $p-\epsilon$ and 1-inputs are accepted with probability at least $p+\epsilon$ for some constant $p>0$ and that uses communication $c$ at most on all inputs. The cost of $\cal P$ is then $c/\epsilon^2$.

   We denote the complexity measure of the clocked one-clean-qubit model by $Q_{[1]}(f)=\inf_{\cal P}\frac{communication({\cal P})}{bias({\cal P})^2}$, where the infimum is over all protocols $\cal P$  for $f$.
\end{definition}

The motivation behind Definition \ref{compmeasure} that it seems unlikely that the success probability can always be amplified arbitrarily. Therefore, we allow the protocol to  run with an arbitrarily bad bias but include the cost that it would take to bring this bias up by a standard amplification (repeat the computation $O\big(\frac{1}{bias^2}\big)$ times): in the situation described in Definition \ref{compmeasure} by a standard Chernoff bound repeating $t=4/\epsilon^2$ times (and accepting if at least $pt$ runs accepted)  would lead to error at most $1/3$.

There is no prior entanglement allowed in this model because the EPR-pairs could be used to create more pure qubits, simply by sending one qubit from one communicating party to another, who can then make the state $\ket{00}$. It is also essential that measurements are performed only at the end of the computation, or a pure state could be obtained by measuring the state (\ref{model}).

In our paper, we allow arbitrary projective measurements in the one-clean-qubit model. There are papers such as \cite{Shor} and \cite{classsim} defining the one-clean-qubit model in a way such that it measures only one qubit at the end of the computation. However, in Theorem \ref{strong}, we show that there is only negligible difference between these definitions in communication complexity.

\subsection{Clocked and Semi-unclocked Models}
There are two types of models being considered: the clocked model and the semi-unclocked model.

\begin{definition}[Clocked model]
   In the clocked model, the message in round $i$ is computed by a unitary that can depend on $i$. In other words, the protocol knows $i$ without having to store $i$ anywhere. The communication channel of a clocked model is ghosted, i.e. different qubits can be communicated in different rounds.
\end{definition}
\begin{figure}[h]
    \centering
    \includegraphics[width=\textwidth, height=3cm]{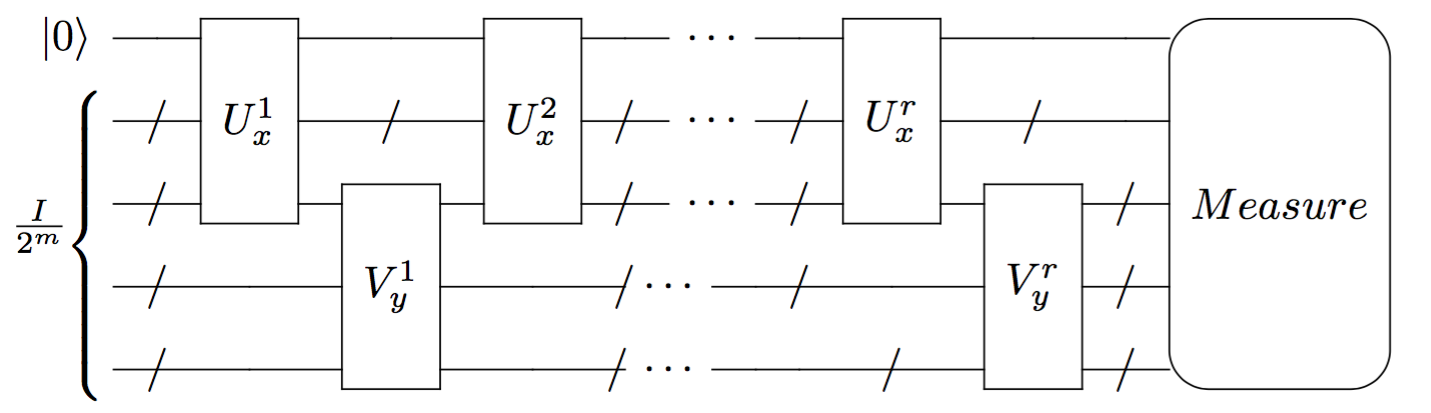}
    \caption{Clocked model}
    \label{fig:my_label}
\end{figure}

Protocols in the clocked model implicitly use a counter to tell the protocol which round it is in. This counter could be considered as extra classical storage, so we define another model that does not allow this. In that model, however, protocols still need to know when to stop, and since no intermediate measurements are allowed, we simply switch the protocol off after the correct number of rounds, and measure.

\begin{definition}[Semi-unclocked model]
   In the semi-unclocked model, the same unitary must be applied in every round. The protocol terminates after a fixed number of rounds. The communication channel of a semi-unclocked model is fixed, i.e., the same qubits have to be communicated in every round.
\end{definition}

\begin{figure}[h]
    \centering
    \includegraphics[width=\textwidth, height=3cm]{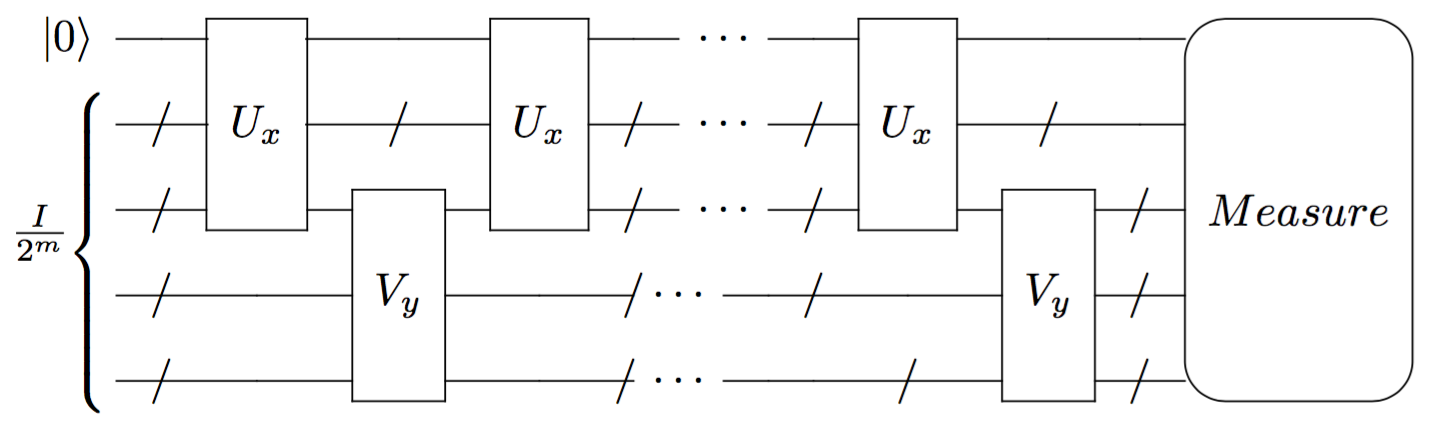}
    \caption{Semi-unclocked model}
    \label{fig:my_label}
\end{figure}

\begin{example}\label{ex:IP}
   The inner product modulo 2 problem is defined as follows: \[IP_2(x, y)=\sum_i x_iy_i\mod 2, \hspace{1mm} \mbox{where} \hspace{1mm} x, y\in\{0, 1\}^n.\] Under the clocked model $\hat{P}$ shown in Figure \ref{IP2}, let $U_x^i$ be Alice's unitary and let $V_y^i$ be Bob's for $i=1\cdots n$. We start with two clean qubits. The first qubit is meant  to store Alice's $x_i$ while the second stores $\sum_i x_iy_i \mod 2$. The protocol (informally)  goes as follows:

   In the first round, Alice stores $x_1$ in the first qubit and sends the two qubits to Bob, who multiplies $x_1$ in the first qubit with his $y_1$ and stores the product in the second qubit. He then sends the first qubit back to Alice. For every round $i=2,\cdots,n$,
   \begin{enumerate}
       \item $U_x^i$ first XORs $\ket{x_{i-1}}$ on the first qubit with $x_{i-1}$, thereby restoring the qubit to $\ket{0}$, before storing the value $x_i$ in it.
       \item Alice sends the first qubit to Bob.
       \item $V_y^i$ multiplies $y_i$ with $x_i$ (stored in the first qubit) and adds the product to the sum stored in the second qubit modulo 2.
       \item Bob sends the first qubit back to Alice.
   \end{enumerate}

   The communication terminates after a total number of $2n-1$ rounds and the bias is $\frac{1}{2}$ (i.e. zero error). Bob does the measurement, the total communication is $2n$.

   \begin{figure}[h]
      \centering
      \includegraphics[width=\textwidth, height=2.5cm]{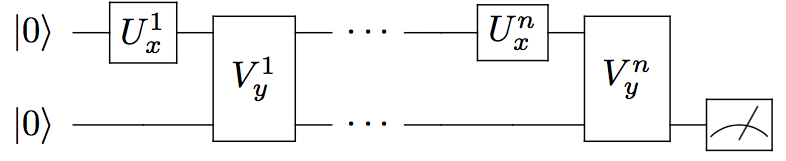}
      \caption{Clocked two-clean-qubit model for computation of inner product modulo 2}
      \label{IP2}
   \end{figure}
   $\hat{P}$ can be simulated with a clocked one-clean-qubit protocol that uses 1 clean qubit and 2 mixed qubits.
   The unitary $\mathcal{M}$ does the following:
   $$\mathcal{M}:
   \begin{cases}
   \ket{0}\otimes\ket{0}\otimes\ket{0}\mapsto\ket{1}\otimes\ket{0}\otimes\ket{0}\\
   \ket{0}\otimes\ket{z_1}\otimes\ket{z_2}\mapsto\ket{0}\otimes\ket{z_1}\otimes\ket{z_2}\\
   \end{cases},$$ where $\ket{z_1}$ or $\ket{z_2}\neq\ket{0}$. Extend to a unitary arbitrarily. In other words, $\mathcal{M}$ flips the first qubit if the next two qubits are both in the $\ket{0}$ state (this happens with probability $\frac{1}{4}$). After applying $\mathcal{M}$, the protocol is carried out as per $\hat{P}$. The measurement is done as follows:
   \begin{itemize}
       \item If the first qubit is $\ket{0}$, a "coin toss" is being performed for the output (e.g.~measure yet another mixed qubit).
       \item If the first qubit is $\ket{1}$, the measurement is done as per $\hat{P}$.
   \end{itemize}
   Note that the two measurements can be combined into one.

   Therefore, we get an error probability of $$\frac{3}{4}\cdot\frac{1}{2}=\frac{3}{8},$$ and a bias of $\frac{1}{8}$.
   The total communication is $2n+1$ and hence the cost is $64(2n+1)=O(n)$.
\end{example}

We now compare the $k$-clean-qubit model with the one-clean-qubit model and also the clocked model with the semi-unclocked model. We prove the following theorems:
\begin{theorem}\label{c-c}
   Given a clocked $k$-clean-qubit protocol $\mathcal{P}$ for a function $f$ that has communication $c$ and a bias of $\epsilon$, there exists a clocked one-clean-qubit protocol $\Tilde{\mathcal{P}}$ for $f$ that has communication $c$ (or $c+1$ depending on which player does the measurement), and a bias of $\frac{\epsilon}{2^k}$ .
\end{theorem}

\begin{theorem}\label{strong}
     Given a clocked $k$-clean-qubit protocol $\Tilde{\mathcal{P}}$ for a function $f:X\times Y\to\{0,1\}$ with a ghosted communication channel, that does an arbitrary projective measurement with two outcomes, has communication $c$ and a bias of $\epsilon$, there exists a semi-unclocked one-clean-qubit protocol $\mathcal{P}_f$ for $f$ with a fixed communication channel, that does a measurement on one qubit, has communication $O(c\log{c})$ and a bias of $\Omega(\frac{\epsilon}{2^k})$.
\end{theorem}

The proofs of Theorems \ref{c-c},\ref{strong} are in the appendix.
Applying Theorem \ref{strong} to Example \ref{ex:IP} gives the following.

\begin{corollary}
The semi-unclocked one-clean-qubit quantum communication complexity of $IP_2$ is $O(n\log n)$.
\end{corollary}

\section{One-way Complexity with One Clean Qubit}\label{ow}
\subsection{The Upper Bound on $Q_{[1]}^{A\rightarrow B}(f)$}
Let $Q_{[1]}^{A\rightarrow B}(f)$ denote the complexity measure of a one-way two-player one-clean-qubit protocol. We define a one-way two-player one-clean-qubit protocol as follows:

\begin{definition}[One-way two-player one-clean-qubit protocol]
    The computation in the one-way version of one-clean-qubit protocols starts with a single qubit in the clean state and the rest of the qubits in the totally mixed state. The first player applies her  unitary on an arbitrary number of qubits, sends some of the qubits to the next player who also applies his unitary on an arbitrary number of qubits, and does a measurement. The cost is defined as for general one-clean qubit protocols. This can be described by the figure below:
    \begin{figure}[h]
       \centering
       \includegraphics[width=8cm, height=3cm]{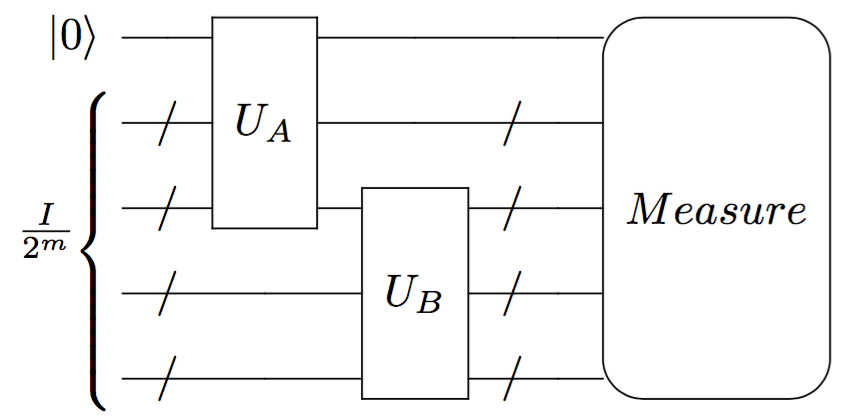}
       \caption{One-round one-clean-qubit protocol}
       \label{fig:my_label}
   \end{figure}

   Note that this type of protocol is semi-unclocked by definition.
\end{definition}

We show an upper bound in terms of the weakly unbounded-error communication complexity.

\begin{definition}[Weakly  unbounded-error  protocol, $PP$]\label{dfpp}
   In a weakly unbounded-error (randomized) protocol ($PP$ protocol),
   the function $f$ is computed correctly with probability greater than  $\frac{1}{2}$ by a classical private coin protocol. The cost of the protocol with a maximum error (over all inputs) of $\frac{1}{2}-\epsilon$ and a maximum communication of  $c$, is given by $PP(f)= c-\floor{\log{\epsilon}}$. \cite{klauck:qcclowerj}
\end{definition}

We show the following theorem for the upper bound on the communication complexity of the one-clean-qubit one-way protocols (in the appendix):
\begin{theorem}\label{theorem1}
   $Q_{[1]}^{A\to B}(f)\leq 2^{O(PP(f))}$.
\end{theorem}

\subsection{The Lower Bound on $Q_{[1]}^{A\rightarrow B}(f)$}


\begin{theorem}\label{th:owl}
   For all $f:\{0,1\}^n\times \{0,1\}^n \to \{0,1\}$ we have $Q_{[1]}^{A\rightarrow B}(f)\geq 2^{\Omega(PP(f))-O(\log{n})}$.
\end{theorem}

The proof relies only on the fact that an efficient one-way one-clean-qubit protocol needs to achieve a large enough bias.
The communication needed to do so is immaterial for our lower bound, which is quite interesting. In other words, there is a threshold to the bias which simply cannot be passed even if we allow more qubits to be sent. This is in sharp contrast to many common modes of communication with error.

The bound on the achievable bias comes from margin complexity, an important concept in learning theory \cite{margin}. The proof is in the appendix.

\section{The Trivial Lower Bound on $Q_{[1]}(f)$}
The lower bound on the two-way one-clean-qubit communication complexity $Q_{[1]}(f)\geq \Omega(Q(f))$ is trivial since one-clean-qubit protocols can be turned into standard quantum protocols at their cost. In Appendix \ref{ap:triv} we discuss this lower bound for some well-known functions.

\section{Hardness of Classically Simulating the One-Clean-Qubit Model}\label{r2}

We now turn to {\em simulations} of quantum protocols with the one-clean-qubit restriction by randomized protocols. The most demanding definition of simulating a quantum protocol by a randomized protocol is that the randomized protocol must replicate the acceptance probabilities of a given quantum protocol on all inputs, up to some additive error\footnote{We only consider additive error.}.

Our weaker definition of an $\epsilon$-error simulation is:
\begin{definition}[$\epsilon$-error simulation of a quantum protocol]
   Given a quantum protocol $\mathcal{P}$ for a function $f:X\times Y\to \{0, 1\}$ such that for all inputs $(x, y)\in X\times Y$ , $\mathcal{P}$ accepts 1-inputs with probability at least $\alpha$ and accepts 0-inputs with probability at most $\beta$. A classical simulation of $\mathcal{P}$ with additive error of $\epsilon$ is one that accepts 1-inputs with probability at least $\alpha -\epsilon$ and accepts 0-inputs with probability at most $\beta +\epsilon$.
\end{definition}
\begin{remark}
   The above definition is nontrivial only if $\alpha-\epsilon > \beta + \epsilon$.
\end{remark}

\subsection{Simulating the One-Clean-Qubit Model with Polynomially Small Additive Error}
We show the following lemma (see the appendix):
\begin{lemma}\label{Lemma1}
   Given any two-round (Alice $\rightarrow$ Bob $\rightarrow$ Alice) $k$-clean-qubit quantum protocol (with communication $2k$ and where both messages contain only the $k$ clean qubits) for a function $f$ that accepts 0-inputs with probability at most $q$ and accepts 1-inputs with probability at least $p$, there exists a two-round one-clean qubit protocol (with communication $2k$) for the same function that accepts 0-inputs with probability at most $\frac{q}{2^k}$ and accepts 1-inputs with probability at least $\frac{p}{2^k}$.
\end{lemma}
\begin{theorem}\label{long}
   In communication complexity, there exists a function $f:\{0, 1\}^n\times\{0, 1\}^n\to\{0, 1\}$ and a one-clean-qubit quantum protocol $\mathcal{P}$ with communication $O(\log{n})$  such that simulating $\mathcal{P}$ classically  with an allowance of $\frac{1}{n^4}$ additive error requires $\Theta(n)$ communication.
\end{theorem}
\begin{proof}
   Consider the function below:
   $$MIDDLE(x, y)=0\Leftrightarrow\displaystyle\sum_i x_i y_i=\frac{n}{2},\hspace{1mm}MIDDLE(x, y)=1\Leftrightarrow\displaystyle\sum_i x_i y_i\neq\frac{n}{2},$$
   where $x, y\in\{0, 1\}^n$. With Lemma \ref{Lemma1} in mind, we design a standard quantum protocol first. We would like to compute the state $\frac{1}{\sqrt{n}}\sum_{i=1}^n(-1)^{x_iy_i}\ket{i}$. This can be done by executing the following quantum protocol $\mathcal{P}$:
\begin{enumerate}
    \item Alice prepares the state $\frac{1}{\sqrt{n}}\sum_{i=1}^n\ket{i}\ket{x_i}$ and sends it to Bob.
    \item Bob applies his unitary, which maps the state he received from Alice to $\frac{1}{\sqrt{n}}(-1)^{x_i y_i}\ket{i}\ket{x_i}$ and sends the result to Alice.
    \item  Alice XORs the last qubit with $x_i$ and then traces out that qubit to obtain $\frac{1}{\sqrt{n}}(-1)^{x_i y_i}\ket{i}$, applies a Hadamard transformation and does a complete measurement in the computational basis. The protocol outputs 1 if it measures the all-zero string and outputs 0 otherwise.
\end{enumerate}
This protocol requires $2\log{n}+2$ communication and uses $\log{n}+1$ clean qubits. Finally, we transform the above protocol into a one-clean-qubit protocol according to Lemma \ref{Lemma1}.

Now we compute the acceptance probabilities of the standard quantum protocol above:
\begin{align}\label{ip}
   \begin{split}
       & \braket{H(\frac{1}{\sqrt n}\sum_{i=1}^n\ket{i}(-1)^{x_iy_i})|\ket{00\cdots 0}}\\
       & = \braket{\frac{1}{\sqrt{n}}\sum_{i=1}^n\ket{i}(-1)^{x_iy_i}|H(\ket{00\cdots 0})}\\
     & = \braket{\frac{1}{\sqrt{n}}\sum_{i=1}^n\ket{i}(-1)^{x_iy_i}|\frac{1}{\sqrt{n}}\sum_{i=1}^n\ket{i}}.
   \end{split}
\end{align}

For the case where $\braket{x, y}=\sum_{i=1}^nx_iy_i=\frac{n}{2}$, we have $\frac{n}{2}$ 0's and $\frac{n}{2}$ 1's among the $x_iy_i$ and hence, (\ref{ip}) for this case equals to zero, which implies that the protocol rejects 0-inputs with certainty.

For the case where $\braket{x, y}=\sum_{i=1}^nx_iy_i=\frac{n}{2}+t$, we have $\frac{n}{2}-t$ 0's and $\frac{n}{2}+t$ 1's and hence, the amplitude from (\ref{ip}) is
$$\frac{1}{\sqrt{n}}\cdot\big(\frac{n}{2}+t-(\frac{n}{2}-t)\big)\frac{1}{\sqrt{n}}=\frac{2t}{n},$$
which implies an acceptance probability of $(\frac{2t}{n})^2=\frac{4t^2}{n^2}$.

Notice that the gap between 0- and 1-inputs is $\frac{4t^2}{n^2}$. Now, simulating $\mathcal{P}$ using only one clean qubit does not change the communication but reduces the acceptance probability of 1-inputs from $\frac{4t^2}{n^2}$ to $\frac{2t^2}{n^3}$ and does not change the acceptance probability of 0-inputs. The gap between the acceptance probability of 0-inputs and 1-inputs is now $\frac{2t^2}{n^3}-0=\frac{2t^2}{n^3}$.

We will focus on the 1-inputs with $t=-1$.

We then show that classically simulating the one-clean-qubit protocol with $\frac{1}{n^4}$ additive error for the function $MIDDLE(x, y)$ requires $\Omega(n)$ communication. For this, we use Razborov's analysis of the rectangle bound for the Disjointness problem\cite{recbound} together with a reduction and the fact that the rectangle bound is not sensitive to acceptance probabilities being small. This shows that any classical protocol that simulates the above quantum protocol within additive error $1/n^4$ needs communication $\Omega(n)$.   Details are in Appendix \ref{app:low}.

   \end{proof}
\subsection{Simulating the One-Clean-Qubit Model with  Constant Additive Error}
Previous results about the hardness of simulating the one-clean-qubit model (in circuit complexity) require the additive simulation error to be exponentially small. In the previous subsection we have shown that in communication complexity additive error $1/poly(n)$ is already enough to give a separation (which is also not based on unproven assumptions). Here we consider pushing this even further: can the one-clean-qubit model be simulated classically with constant additive error?

Showing hardness of a classical simulation with constant additive error is equivalent to showing a separation between $Q_{[1]}(f)$ and $R(f)$: regarding both complexity measures efficient error reduction is possible\footnote{We defined $Q_{[1]}$ so.}. And showing hardness of a simulation of a quantum protocol for $f$ within a small constant error means showing $R(f)$ is large.

 The strength of the one-clean-qubit model is trace-estimation. Any communication-like unitary can have its trace estimated by a quantum protocol with only one clean qubit (compare the proof of Theorem \ref{strong}). So we look for a hard problem along those lines. A one-way quantum protocol is not a good choice, since the trace of the product of unitaries applied by Alice and Bob is a vector inner product and can be estimated well by known randomized protocols with small error, if the gap of acceptance between one-inputs and zero-inputs is large \cite{knr:rand1round}. So we look beyond protocols with one round.

For technical reasons (cyclic property of matrix trace), looking for the simplest problem that should exhibit a separation
 we consider the three-player number-in-hand model\footnote{In the three-player number-in-hand model, each player sees only their own input.}.

We conjecture the following:
\begin{conj}
       There exists a function $f$ and a one-clean-qubit quantum protocol $\mathcal{P}$ that computes $f$ exactly with communication $O(\log{n})$  such that simulating $\mathcal{P}$ classically  with an allowance of constant  additive error requires $\Omega(\sqrt{n})$ communication.
\end{conj}

Consider the number-in-hand  $ABC$ problem involving three parties: Alice, Bob and Charlie, who are each given $n\times n$ matrices $A$, $B$ and $C$ respectively, where $A, B, C\in O_n$, where $O_n$ is the orthogonal group. The $ABC$ problem is described by the following function: $$ABC(A, B, C)=1\iff ABC=I,\hspace{5mm}ABC(A, B, C)=0\iff ABC=-I.$$

There is a one-clean-qubit quantum protocol of $O(\log{n})$ communication that accepts 1-inputs and rejects 0-inputs with certainty. The initial state starts off with one qubit in a pure state $\ket{0}$ and $\log{n}$ totally mixed qubits. The protocol goes as follows:
\begin{enumerate}
    \item Alice applies a Hadamard transformation to the clean qubit and obtains $\sigma=H\ket{0}=\frac{1}{\sqrt{2}}(\ket{0}+\ket{1})$. She then tensors it with an arbitrary state $\rho$ on $\log{n}$ qubits (for example $\frac{I}{n}$) and we denote the resulting state as $\zeta$. She then applies her controlled-$A$ unitary to $\zeta$ and gets $\zeta^\prime$. Alice send $\zeta^\prime$ to Bob.
    \item Bob applies his controlled-$B$ unitary to $\zeta^\prime$ and gets $\zeta^{\prime\prime}$. Bob sends $\zeta^{\prime\prime}$ to Charlie.
    \item Charlie applies his controlled-$C$ unitary to $\zeta^{\prime\prime}$ and gets  $\zeta^{\prime\prime\prime}$. He then applies a Hadamard transformation to the first qubit in $\zeta^{\prime\prime\prime}$ and does a measurement.\\
\end{enumerate}
The protocol is illustrated in Figure \ref{ABC}.\\

\begin{figure}[h]
    \centering
    \includegraphics[width=\textwidth]{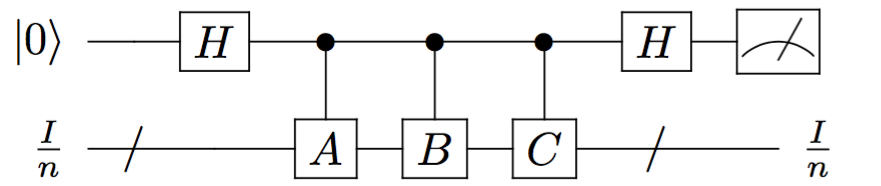}
    \caption{One-clean-qubit protocol for $ABC$}
    \label{ABC}
\end{figure}
\underline{Case 1: if $ABC=I$}\\
The composite of the controlled $A$, $B$ and $C$ is the same as that of a controlled-identity unitary, which does nothing to $\sigma$. When $\sigma$ undergoes a Hadamard transformation before being measured, it becomes the $\ket{0}$ state. The protocol outputs 1 if it measures $\ket{0}$ . \\

\underline{Case 2: if $ABC=-I$}\\
The composite of the controlled unitaries $A$, $B$ and $C$ is similar to that of a controlled-$Z$ unitary, which does a phase flip on $\ket{1}$ in $\sigma$, changing it into $\frac{1}{\sqrt{2}}(\ket{0}-\ket{1})$. We denote the phase-flipped $\sigma$ as $\sigma^\prime$. When $\sigma^\prime$ undergoes a Hadamard transformation before being measured, it becomes the $\ket{1}$ state. The protocol outputs 0 if it measures $\ket{1}$. \\
\begin{lemma}
   There exists a three-player number-in-hand one-clean-qubit protocol that solves $ABC$ exactly with communication $O(\log{n})$.
\end{lemma}

Note that the quantum protocol uses the arbitrary state $\rho$ (here $\rho=I/n$) as a catalyst as in \cite{Buhrman}.
Regarding the randomized complexity of $ABC$, we prove the following theorem:
\begin{theorem}\label{thm:ABC}
   $R(ABC)\leq O(\sqrt{n})$.\end{theorem}

We explain the proof in Appendix \ref{ABC}. Let us note here that due to the cyclic property of the trace both the quantum and classical protocols for ABC are one-way and can be run in any order among the players, e.g.~Charlie to Alice to Bob.

    It remains an open problem to derive a matching lower bound for the randomized communication complexity of $ABC$.

    \begin{conj}\label{open}
       $R(ABC)\geq\Omega(\sqrt{n})$ as long as $n$ is even.    \end{conj}

We now consider a geometric conjecture that implies Conjecture \ref{open}. This conjecture says that if we take two sufficiently large subsets of $SO_n$ (the special orthogonal group), choose two operators independently from them, and multiply them, we get something similar to the uniform distribution on all of $SO_n$.

\begin{conj}\label{geo}
There are constants $\delta>0, \gamma>1$ such that the following is true:

Let $M,R\subseteq SO_n$ and, for the Haar measure $\mu$ on $SO_n$, let $\mu(M),\mu(R)\geq 2^{-\delta\sqrt n}$.
Denote by $\tau$ the density function of the probability distribution that arises, when $B\in M$ and $C\in R$ are chosen uniformly from these sets independently, and the matrix product $BC$ is formed. Then
\[Prob_{A\in SO_n} (\tau(A)\not\in[1/\gamma,\gamma])\leq 2^{-\delta\sqrt n}.\]
\end{conj}

Conjecture \ref{open} follows from Conjecture \ref{geo} by an application of the rectangle bound from communication complexity: A large rectangle/box $L\times M\times R$, where $L,M,R\subseteq SO_n$ leads to a $\tau$ that is similar to the uniform distribution. Only an exponentially small subset of matrices $A\in SO_n$ has $\tau(A)$ not constant. This also implies that $E_{A\in L} \tau(A)=\Theta(1)$, if we throw out the small subset of $A\in L$ where $\tau(A)$ is too large (this does not affect size or error much.)
Denote by $\beta_C$ the density function of the distribution where a random $B\in M$ is multiplied to a fixed $C$. $\tau(A)=E_{C\in R}\beta_C(A^*)$.


Define $H=\{(A,B,C): A,B,C\in SO_n$ and $ABC=I\}$ and $G=\{(A,B,C): A,B,C\in SO_n$ and $ ABC=-I$.
It is easy to show that \[E_{A\in L} E_{C\in R} [\beta_C(A)]=\frac{\mu(L\times M\times R|H)}{\mu(L\times M\times R)}.\]
That means that $\mu(L\times M\times R|H)$ and $\mu(L\times M\times R|G)$ differ by at most a constant factor and $L\times M\times R$ has constant error under the distribution that puts weight 1/2 on each of $G,H$. Hence the rectangle/box $L\times M\times R $ has large error. We use that $n$ is even because otherwise $-I\not\in SO_n$.

Furthermore in the case of odd $n$ Alice, Bob, and Charlie can simply compute $det(ABC)=det(A)det(B)det(C)$ in order to determine whether $ABC=I$ or $ABC=-I$. This does not work in the case of even $n$ of course.

We also note that the corresponding conjecture is wrong for $O_n$, since $SO_n$ is a subgroup of size 1/2 that serves as a counterexample. Note that $SO_n$ does not have any proper subgroups of size larger than 0. This follows from the fact that unlike $O_n$, $SO_n$ has no subgroup that has the same Lie-algebra as itself.

As weaker conjecture, in which the stated probability is upper bounded by a small constant would be sufficient to give a lower bound on one-way protocols and might be much easier to achieve. We note that in \cite{KLim} we have recently shown a lower bound for the related $aBc$ problem, in which Alice and Charlie receive vectors from the sphere instead of matrices in a generalized one-way setting. Note that the protocol for Theorem \ref{thm:ABC} really solves the $aBc$ problem.



\section{Conclusion}
We investigate a communication complexity model in which all storage consist initially of only one clean qubit plus other qubits that start in the totally mixed state, and  where only one projective measurement can be done in the end. Since error reduction is not possible efficiently in this model we define an appropriate complexity measure depending on the bias.

We introduce the notions of clocked protocols with ghosted communication channel and semi-unclocked  protocols with fixed communication channel for this model. Efficient simulations of clocked $k$-clean-qubits protocols by clocked one-clean-qubit protocols as well as simulations of clocked $k$-clean-qubit protocols by semi-unclocked one-clean-qubit protocols are described. Remarkably, the semi-unclocked model is only less efficient by a logarithmic factor compared to the clocked model.


We study one-way protocols in the model and are able to almost pinpoint their complexity in terms of PP-communication complexity: $2^{\Omega(PP(f))-O(\log{n})}\leq Q_{[1]}^{A\to B}(f)\leq 2^{O(PP(f))}$, implying that functions when computed using the one-clean-qubit model have a cost of at most $2^{O(m)}$, where $m$ is the input length, and that this is tight for some functions (one-way).

Classically simulating  a certain one-clean-qubit protocol for the $MIDDLE(x,y)$ problem with $\frac{1}{poly(n)}$ additive error is hard, as a classical simulation with such error requires $\Theta(n)$ communication, compared to the $O(\log{n})$ communication of the one-clean-qubit protocol.

We conjecture that classically simulating the one-clean-qubit protocol we give for the three-player number-in-hand $ABC$ problem with constant additive error requires $\Omega(\sqrt{n})$ communication, compare to the $O(\log{n})$ communication in the one-clean-qubit protocol. We show the corresponding upper bound on $R(ABC)$.

\newpage

\bibliography{qc}
\bibliographystyle{plainurl}


\begin{appendix}

\section{Open Problems}
\begin{itemize}
    \item Prove Conjecture \ref{open} or the weaker version mentioned above that establishes a lower bound for one-way protocols.
    \item What are some nontrivial lower bounds on $Q_{[1]}(f)$, for instance what are $Q_{[1]}(DISJ)$ and $Q_{[1]}(ViS)$? We conjecture that $Q_{[1]}(DISJ)=\Omega(n)$ based on the difficulty of trying to compute the function in the one-clean-qubit model. Suppose that $ViS$ can computed in the one-clean-qubit communication model efficiently (say with $poly(\log)$ communication), then arbitrary one-way quantum protocols can be simulated with low communication in the one-clean-qubit model. However, we assume that such a supposition seems unlikely and hence we conjecture that $Q_{[1]}(ViS)$ is fairly large, possibly even $Q_{[1]}(ViS)=\Omega(n)$.
    \item Is $Q_{[1]}(f)>n$ for any function? A candidate for this problem would be a random function chosen from all functions $f:\{0, 1\}^n\times\{0, 1\}^n\to\{0, 1\}$. It would be interesting if the one-clean-qubit model can compute all or most $f:\{0, 1\}^n\times\{0, 1\}^n\to\{0, 1\}$ with linear cost.
    \item  What are some examples of functions in which $Q_{[1]}(f)>>R(f)$ or $Q_{[1]}(f)<<R(f)$? For instance, for the two-player $ABC$ problem, $ABC_2$, described as follows:
    $$ABC_2(A_1, A_2, B_1, B_2)=1\Leftrightarrow A_1B_1A_2B_2=I,$$
    $$ABC_2(A_1, A_2, B_1, B_2)=0\Leftrightarrow A_1B_1A_2B_2=-I,$$
    where $A_1, A_2$ are Alice's unitaries and   $B_1, B_2$ are Bob's unitaries, $Q_{[1]}(ABC_2)=O(\log{n})$. What is $R(ABC_2)$?
    \item Are there any specific lower bound methods for the semi-unclocked one-clean-qubit protocol?

\end{itemize}

\section{Proof of Theorem \ref{c-c}}

   The clocked $k$-clean-qubit protocol $\mathcal{P}$ illustrated in Figure \ref{clockedk} has communication $c$ and a bias of $\epsilon$. Hence, it has an error probability of $\frac{1}{2}-\epsilon$ and cost $\frac{c}{\epsilon^2}$. Denote by $U^i_x$ Alice's unitaries and by $V^i_y$ Bob's unitaries for $i=1,\cdots, r$.
   Note that $U^i_x$ is defined as a unitary on all qubits, but acts only on Alice's qubits.

   \begin{figure}[h]
       \centering
       \includegraphics[width=\textwidth, height=3cm]{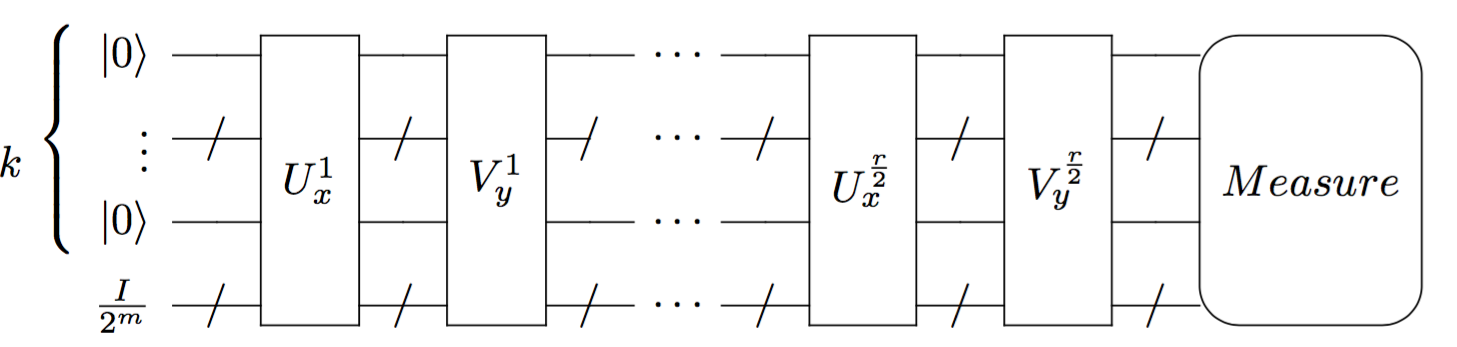}
       \caption{Clocked k-clean-qubits protocol $\mathcal{P}$}
       \label{clockedk}
   \end{figure}

   $\mathcal{P}$ can be modified into a clocked one-clean-qubit protocol $\Tilde{\mathcal{P}}$ as in Figure \ref{clockedone} with about the same amount of communication.
   \begin{figure}[h]
       \centering
       \includegraphics[width=\textwidth, height=2cm]{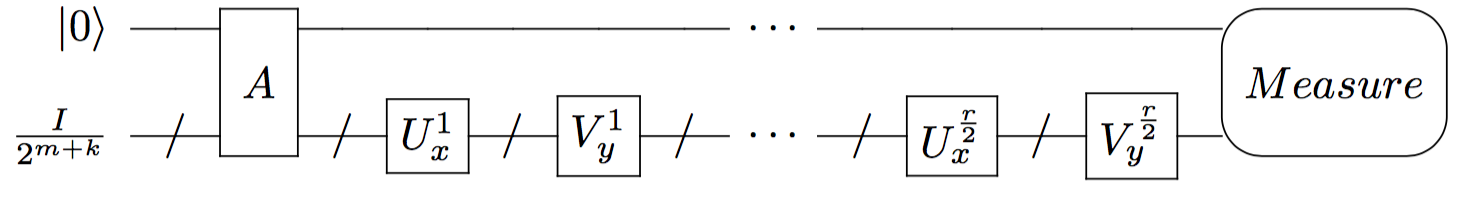}
       \caption{Clocked one-clean-qubit model $\Tilde{\mathcal{P}}$}
       \label{clockedone}
   \end{figure}

   In $\Tilde{\mathcal{P}}$, the unitary $A$ does a bit flip on the first qubit if the next $k$ qubits are in the $\ket{0}$ state, and does nothing otherwise. All the $k+m$ mixed qubits undergo the same series of unitary transformation as  in $\mathcal{P}$. The measurement in $\Tilde{\mathcal{P}}$ is done as follows:
   \begin{itemize}
       \item If the first qubit is $\ket{0}$, a "coin toss" is being done.
       \item If the first qubit is $\ket{1}$, the measurement is carried out as per $\mathcal{P}$.
   \end{itemize}
   Note that the two measurements can be combined into one.

   The communication in $\Tilde{\mathcal{P}}$ is $c$ or $c+1$, depending on which player does the measurement. If the measurement is done by the player who begins the communication, the communication is $c$. Otherwise, the first qubit has to be sent to the other player for the measurement to be done, causing the communication to be increased to $c+1$.

  The error probability of $\Tilde{\mathcal{P}}$ can be computed to be $$(1-\frac{1}{2^k})\cdot \frac{1}{2}+\frac{1}{2^k}\cdot(\frac{1}{2}-\epsilon)=\frac{1}{2}-\frac{\epsilon}{2^k}.$$ Hence, the bias decreases from $\epsilon$ to $\frac{\epsilon}{2^k}$.

  The cost of $\Tilde{\mathcal{P}}$ is given by $c\cdot(\frac{2^k}{\epsilon})^2$ or $(c+1)\cdot(\frac{2^k}{\epsilon})^2$.

\section{Proof of Theorem \ref{strong}}

From Theorem \ref{c-c}, a clocked $k$-clean qubit protocol $\Tilde{\mathcal{P}}$ with a ghosted communication channel that does an arbitrary projective measurement and has communication $c$ and a bias of $\epsilon$, can be modified into a clocked one-clean-qubit protocol $\mathcal{P}$ with a ghosted communication channel, that does an arbitrary projective measurement, has communication $c+1$ and bias $\frac{\epsilon}{2^k}$. The total number of qubits is $m+k+1$, with 1 clean qubit.

   We would like to turn $\mathcal{P}$ into a protocol $\mathcal{P}^\prime$ that measures only one qubit in the computational basis. This can be done by adding an extra clean qubit and replacing the measurement in $\mathcal{P}$ with a unitary operator $U_S$ and a measurement that measures the newly added qubit in the standard basis. $U_S$ does the following:
   $$U_S:\begin{cases}
   \ket{a}\ket{b_i}\mapsto\ket{a}\ket{b_i}, \text{for}\hspace{1mm}b_i\in B\\
   \ket{a}\ket{b_i}\mapsto\ket{a\oplus 1}\ket{b_i}, \text{for}\hspace{1mm}b_i\notin B\\
   \end{cases},$$ where $a\in\{0, 1\}$ and $B=\{b_1,\cdots, b_l\}$ is the basis of the subspace  $S\subseteq\mathbb{C}^{m+k+1}$, which is a constituent of the observable used to measure the quantum state in $\mathcal{P}$.

    In other words, $U_S$ flips the first qubit  on any basis vector $b_i\notin B$, and does nothing otherwise. The resulting protocol $\mathcal{P}^\prime$ is as follows:

   \begin{figure}[h]
      \centering
      \includegraphics[width=0.5\textwidth, height=2.5cm]{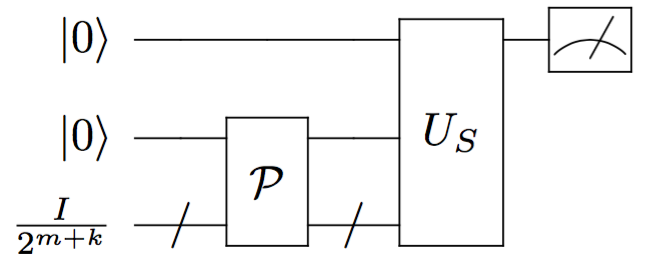}
      \caption{Clocked two-clean-qubit protocol that measures one qubit $\mathcal{P}^\prime$}
      \label{arbtoone}
   \end{figure}

   \begin{remark}
      A clocked protocol with a ghosted communication channel can be easily converted to one with fixed channel in which Alice and Bob take turns to send one qubit each. This at most doubles the communication.
   \end{remark}

   In the new protocol, the communication channel is fixed, the total communication is increased to at most $2(c+1)$,  and the bias remains unchanged.

   According to Shor \cite{Shor}, the probability of measuring 0 (which corresponds to acceptance) can be made to depend only on  the trace of  a unitary operator as shown below.

   Consider the following trace estimation protocol $\mathcal{P}_{main}$ illustrated in Figure \ref{pmain},
   \begin{figure}[h]
      \centering
      \includegraphics[width=0.6\textwidth, height=2cm]{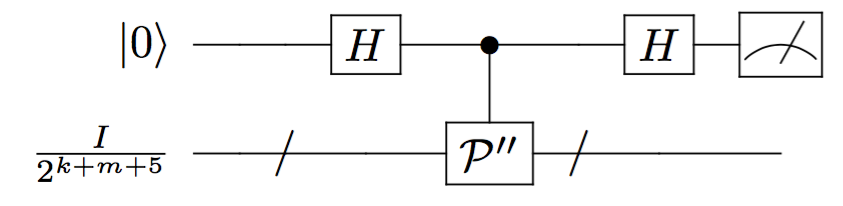}
      \caption{Trace estimation protocol $\mathcal{P}_{main}$}
      \label{pmain}
   \end{figure}

   which contains the unitary operator  $\mathcal{P}^{\prime\prime}$ shown in Figure \ref{pprime}.
$\mathcal{P}_{main}$ accepts with probability \[\frac12+\frac{Re(Tr(\mathcal{P''}))}{2^{d+1}},\]
 where $d=m+k+5$ is the number of qubits in $\mathcal{P}''$ and $Re(x)$ is the real part of $x$.

   \begin{figure}[h]
      \centering
      \includegraphics[width=0.8\textwidth, height=5cm]{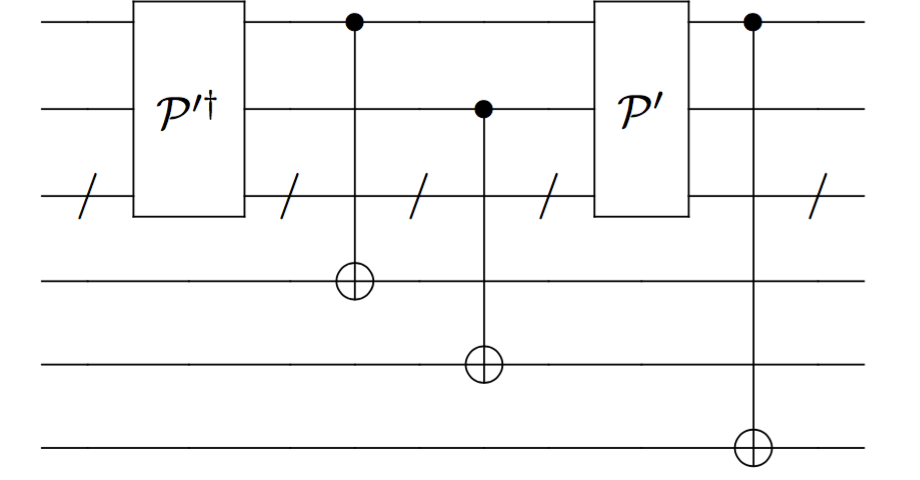}
      \caption{$\mathcal{P}^{\prime\prime}$}
      \label{pprime}
   \end{figure}

Let $I_\ell$ denote the $2^{\ell}$-dimensional identity matrix.
We have that \[Tr[(\ket{0}\bra{0}\otimes I_{m+k+1})\mathcal{P}^\prime(\ket{0}^2\bra{0}^2\otimes I_{m+k})\mathcal{P}^{\prime\dagger}]=\frac18 Tr[\mathcal{P}''],\]
because $Tr[\mathcal{P}'']=\sum_{x\in \{0,1\}^{m+k+5}} \bra{x}\mathcal{P''}\ket{x}$, and so for instance basis vectors $\ket{x}$ that have a 1 in qubit 1 contribute nothing to the sum due to the rightmost CNOT.
Similarly, the other CNOTs correspond to the other projection one the left hand side.
This equation also shows that the right-hand-side trace is real: up to scaling the left hand side corresponds to a probability of measuring $0$ when running $\mathcal{P}'$ on the two-clean-qubit state.

The acceptance probability of $\mathcal{P}_{main}$ is given by

\begin{align*}
\begin{split}
         p_0 & = \frac{1}{2}+ \frac{Tr[\mathcal{P}^{\prime\prime}]}{2^{k+m+6}}\\
         & = \frac{1}{2}+\frac{8\cdot Tr[(\ket{0}\bra{0}\otimes I_{m+k+1})\mathcal{P}^\prime(\ket{0}^2\bra{0}^2\otimes I_{m+k})\mathcal{P}^{\prime\dagger}]}{2^{k+m+6}}\\
         & = \frac{1}{2}+\frac{8\cdot 2^{k+m}\cdot(\frac{1}{2}+\frac{\epsilon}{2^k})}{2^{k+m+6}}\\
         & = \frac{1}{2}+ \frac{1}{16}+\frac{\epsilon}{2^{k+3}}\\
      \end{split}
   \end{align*}

   \begin{remark}
      The factor of 8 instead of 4 as in \cite{Shor} is due to the presence of three CNOT gates/extra qubits instead of two.
   \end{remark}

   The communication of $\mathcal{P}_{main}$ is four times the communication of $\mathcal{P}'$, since $\mathcal{P}''$ runs $\mathcal{P}'$ backwards and forwards, and because the clean control qubit in $\mathcal{P}_{main}$ must be communicated in every round (every round communicates only one qubit in $\mathcal{P}'$), i.e. the communication becomes $8(c+1)$. The bias decreases to $\frac{\epsilon}{2^{k+3}}$ and is around $\frac{1}{2}+\frac{1}{16}$ instead of $\frac{1}{2}$.

   Lastly, we turn $\mathcal{P}_{main}$ into a semi-unclocked protocol $\mathcal{P}_f$ by adding $\log{r}$ mixed qubits to act as a counter, where $r$ is the number of rounds.
   The resulting protocol looks as follows:

   \begin{figure}[h]
       \centering
       \includegraphics[width=\textwidth, height=3cm]{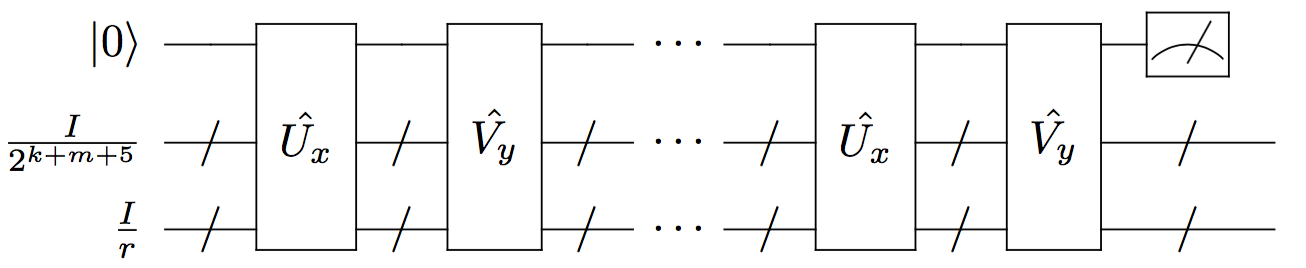}
       \caption{Semi-unclocked one-clean-qubit protocol that measure one qubit $\mathcal{P}_f$}
       \label{semiunclockedpmain}
   \end{figure}

   In $\mathcal{P}_f$, $\hat{U_x}=(H\otimes I)\cdot U_x\cdot (H\otimes I)$, where
   \[U_x:\ket{z}\ket{i}\mapsto(U_x^i\ket{z})\ket{i}\]

   and $\hat{V_y}=(H\otimes I)\cdot V_y\cdot (H\otimes I)$, where
   \[V_y:\ket{z}\ket{i}\mapsto(V_y^i\ket{z})\ket{i+1\mod r},\]
   for all $z\in\{0, 1\}^{k+m+6}$, for all $i\in\{0, 1\}^{\log{r}}$ and where $U_x^i$ and $V_y^i$ are the unitaries from $\mathcal{P}^{\prime\prime}$.

   This means that, starting from a random $j$ on the counter, the unitaries $\hat{V_y}$ and $\hat{U_x}$ apply $V_y^i$ and $U_x^i$ in the correct, but shifted order. Also note that the Hadamard operators cancel out in between consecutive unitaries, and only the first and last have an effect.

   \begin{fact} [Cyclic property of matrix trace]\label{cyclic}
      The trace of a product of three or more square matrices is invariant under cyclic permutations of the order of multiplication of the matrices.
   \end{fact}

   Since the acceptance probability of $\mathcal{P}_{main}$ depends only on the trace of the product of the sequence of unitary operators in $\cal P''$, it follows from Fact \ref{cyclic} that the counter can start from any arbitrary $j\mod r$ without affecting the acceptance probability of $\mathcal{P}_f$.

   The protocol terminates after $r$ rounds of communication. Note that $r=\Theta(c)$, the total communication is now $8(c+1)+O(c\log{c})=O(c\log{c})$. The bias is remains unchanged from $\mathcal{P}_{main}$, i.e. $\Omega(\frac{\epsilon}{2^k})$.

   \section{Proofs Concerning One-Way Protocols}

   \subsection{Proof of Theorem \ref{theorem1}}

     Consider a $c$-bit $PP$-communication protocol $\mathcal{P}$ with bias $\epsilon$ where Alice sends a message $T(x)$ of length $c$ to Bob.
   \begin{enumerate}
       \item We define Alice's unitary $U_A^x$ such that
       \begin{itemize}
           \item If $z=T(x)$, then $U_A^x:\ket{0}\ket{z_1\cdots z_{c}}\mapsto\ket{1}\ket{z_1\cdots z_{c}}$
           \item If $z\neq T(x)$, then $U_A^x:\ket{0}\ket{z_1\cdots z_{c}}\mapsto\ket{0}\ket{z_1\cdots z_{c}}$
        \end{itemize}
         and extend to a unitary in any possible way, for all $z\in\{0, 1\}^c$. Alice applies $U_A^x$ to the initial state, and computes $U_A^x(\ket{0}\bra{0}\otimes\frac{I}{2^c})U_A^{x\dagger}$.
       \item Alice then sends the result $\sigma$ to Bob. This requires $c+1$ qubits of communication.
       \item Upon receiving $\sigma$ from Alice, Bob tensors it with $\frac{I}{2}$ and obtains the state $\sigma\otimes\frac{I}{2}$.
       Bob then applies the unitary $V_B^y$ to the state $\sigma\otimes\frac{I}{2}$, in particular, $V_B^y(\sigma\otimes\frac{I}{2})V_B^{y*}$, as follows
       \[V_B^y:
       \begin{cases}
          \ket{0}\ket{z_1\cdots z_{c+1}}\mapsto\ket{0}\ket{z_{c+1}}\ket{z_1\cdots\ z_{c}}\\
          \ket{1}\ket{z_1\cdots z_{c+1}}\mapsto\ket{1}U_B^y\otimes I\ket{z_1\cdots z_{c+1}},\\
      \end{cases}
      \]
      for all $z\in\{0, 1\}^{c+1}$. That is to say, if the first qubit of $\sigma\otimes\frac{I}{2}$ is 1, $V_B^y$ will apply the protocol unitary $U_B^y$. Otherwise, a "coin toss" is done by flipping the last qubit over to the second position.
      \item Lastly, he does the measurement on the second qubit.
   \end{enumerate}

   The probability of the correct message is $\frac{1}{2^c}$. With a protocol of bias $\epsilon$ (and hence and error of $\frac{1}{2}-\epsilon$), the acceptance probability of the message is $\frac{1}{2^c}(\frac{1}{2}+\epsilon)$. On the other hand, the acceptance probability of the message in the "coin toss" is given by $\frac{1}{2}(1-\frac{1}{2^c})$. Therefore, we have the total acceptance probability:
   \begin{equation}\label{bias}
       \begin{split}
           (1-\frac{1}{2^c})\frac{1}{2}+\frac{1}{2^c}(\frac{1}{2}+\epsilon)
           & = \frac{1}{2}-\frac{1}{2^{c+1}}+\frac{1}{2^{c+1}}+\frac{\epsilon}{2^c}\\
           & = \frac{1}{2}+\frac{\epsilon}{2^c}\\
       \end{split}
   \end{equation}
   The total cost of the protocol is bounded as follows:
   \begin{equation}\label{Q_{[1]}(f)}
       Q_{[1]}^{A\to B}(f) \leq (c+1) \cdot\frac{1}{\epsilon^{\prime 2}}=(c+1)\cdot 2^{2c} \cdot\frac{1}{\epsilon^{2}}\leq 2^{2PP(f)}\cdot({PP(f)}+1)\leq 2^{O(PP(f))},
   \end{equation}
   where $\epsilon^{\prime} = \frac{\epsilon}{2^c}$ from (\ref{bias}).

   \subsection{Proof of Theorem \ref{th:owl}   }

Before we delve into the proof we need a few definitions.We define the notion of  rectangles and two complexity measures: discrepancy and margin complexity.

\begin{definition}[Rectangle]
    A rectangle in $X\times Y$ is a subset $R\subseteq X\times Y$ such that $R=A\times B$ for some $A\subseteq X$ and $B\subseteq Y$.
\end{definition}

\begin{definition}[Discrepancy]
   Let $f:X\times Y\to\{0, 1\}$ be a function, $R$ be any rectangle in the communication matrix, and $\mu$ be a probability distribution on $X\times Y$. The discrepancy of $f$ according to $\mu$ is
   \[disc_{\mu}(f)=\max_R|\Pr_\mu[f(x, y)=0 \hspace{1mm} and \hspace{1mm} (x, y)\in R]-\Pr_\mu[f(x, y)=1 \hspace{1mm} and \hspace{1mm} (x, y)\in R]|.\] Denote $disc(f)=\min_\mu disc_\mu (f)$ as the discrepancy of $f$ over all distributions $\mu$ on $X\times Y$.
 \end{definition}
 It is know that $PP(f)\geq \Omega(\log({\frac{1}{disc(f)}}))$ from Fact 2.8 in \cite{klauck:qcclowerj}, and from Theorem 8.1 in \cite{klauck:qcclowerj} we get $PP(f)\leq O(\log({\frac{1}{disc(f)}}))+\log{n})$.

 \begin{definition}[Margin \cite{margin}]
     For a function $f:X\times Y\to\{0, 1\}$, let $M_f$ denote the sign matrix where all entries are $M_f(x,y)=(-1)^{f(x,y)}$. The margin of $M_f$ is given by: $$m(M_f)=\sup_{\{a_x\}, \{b_y\}} \min_{x, y} \frac{|\braket{a_x|b_y}|}{||a_x||_2||b_y||_2},$$ where the supremum is over all systems of vectors (of any length) $\{a_x\}_{x\in X}, \{b_y\}_{y\in Y}$ such that $sign(\braket{a_x|b_y})=M_f(x,y)$ for all $x,y$.
\end{definition}

 The notion of margin complexity determines the extent to which a given class of functions can be learned by large margin classifiers, which is an important class of machine learning algorithms \cite{margin}.

\begin{proof}
    Assume that the protocol measures the first qubit in the computational basis (if not, then a similar construction as in Theorem \ref{strong} can be used to make this true). The probability of measuring zero is given by $\frac{1}{2}+\frac{tr(I_A\otimes U_B^y\cdot U_A^x\otimes I_B)}{2^{m+1}}$ \cite{Shor}, where $m$ is the total number of qubits involved and the bias is the term $\frac{tr(I_A\otimes U_B^y\cdot U_A^x\otimes I_B)}{2^{m+1}}$. Note that $I_A$ and $I_B$ act on the private qubits of Alice and Bob respectively. Let $ U_A^x \otimes I_B=A_x$ and $I_A\otimes U_B^y=B_y$, and it follows that $$\frac{tr(I_A\otimes U_B^y\cdot U_A^x\otimes I_B)}{2^{m+1}}=\frac{tr(B_yA_x)}{2^{m+1}}=\frac{\braket{b_y|a_x^T}}{2^{m+1}}=\frac{\braket{b_y|a_x^T}}{2||a_x||_2||b_y||_2},$$ where $a_x$ and $b_y$ are the matrices $A_x$ and $B_y$ viewed as vectors, since $A_x$ and $B_y$ are unitary and hence $||a_x||_2=||b_y||_2=2^{\frac{m}{2}}$. If the protocol has bias $\epsilon$, then  $\frac{\braket{b_y|a_x^T}}{2^{m+1}}\geq\epsilon$ for $f(x, y)=1$ and $\frac{\braket{b_y|a_x ^T}}{2^{m+1}}\leq -\epsilon$ for $f(x, y)=0$.

    \begin{remark}
        The size of the unitary matrices does not matter, which is good, since there can be an arbitrarily number of private qubits used by the players but never communicated.
    \end{remark}

    We know from the above that the best possible bias satisfies  $2\epsilon\leq m(f)$. From Theorem 3.1 in \cite{margin} which states that $disc(A)=\Theta(m(A))$, and from Theorem 8.1 in \cite{klauck:qcclowerj}, which states that $PP(f)\leq O(-\log disc(f)+\log n)$ we have $$Q_{[1]}^{A\rightarrow B}(f)\geq\frac{4}{m^2(f)}\geq 2^{\Omega(PP(f))-O(\log{n})}.$$
\end{proof}

\begin{remark}
   This lower bound holds regardless of how much communication is involved: it follows from the fact that one-way one-clean qubit protocols cannot achieve a better bias.
\end{remark}

   \section{The Trivial Lower Bound}\label{ap:triv}

   $Q(f)$ for some functions is given as below \cite{EQ, DISJ, cdnt:ip, ViS, INDEX}:
\begin{itemize}
    \item The equality function (EQ) defined as $$EQ(x, y)=1\iff x=y \hspace{5mm}, \hspace{5mm} EQ(x, y)=0\iff x\neq y,$$ where $x, y\in\{0, 1\}^n$, has $Q(EQ) = \Theta(\log n)$. \\
    Note: No public coin or entanglement.
    \item The disjointness function (DISJ) defined as
    $$DISJ(x, y)=1\iff x\cap y=\emptyset\hspace{5mm},\hspace{5mm }DISJ(x, y)=0\iff x\cap y\neq\emptyset,$$ where $x, y\in\{0, 1\}^n$, has $Q(DISJ) = \Theta(\sqrt{n})$. \\
    Note: $\Omega(\sqrt{n})\leq Q_{[1]}(DISJ)\leq O(n)$.
    \item The inner product modulo two function ($IP_2$) defined as $$IP_2(x, y)=\sum_{i}x_iy_i \mod 2,$$ where $x, y\in\{0, 1\}^n$, has $Q(IP_2) = \Theta(n)$.\\
    Note: $Q_{[1]}^{A\to B}(IP_2)=2^{\Theta(n)}$ while $Q_{[1]}(IP_2)=\Theta(n)$.
    \item The vector in subspace function (ViS) defined as $$ViS(v, H_0)=1\iff v\in H_0\hspace{5mm}, \hspace{5mm}ViS(v, H_0)=0\iff v\in H_0^{\bot},$$ where $v\in\mathbb{R}^n$ and $H_0\subseteq\mathbb{R}^n$ is a subspace with dimension $\frac{n}{2}$, has $Q(ViS) = \Theta(\log n)$.
    \item The index function (INDEX) defined as $$INDEX(x, i)=x_i,$$ where $x\in\{0, 1\}^n$ and $1\leq i\leq n$ has $Q(INDEX) = \Theta(\log{n})$.
\end{itemize}
$Q_{[1]}(EQ)$, $Q_{[1]}(ViS)$ and $Q_{[1]}(INDEX)$ are basically unknown: the lower bounds we know are  $\Omega(\log n)$, but the upper bounds we have are $O(n)$ for $INDEX$ and $EQ$, while Theorem \ref{c-c} implies $Q_{[1]}(ViS)=O(n^2\log n)$.

\section{Proof of Lemma \ref{Lemma1}}

   In the quantum protocol, Alice prepares the first message $\ket{\phi_x}$ by applying a protocol unitary $W_x^{(1)}$ to the all-zero state on the $k$ clean qubits, and sends it to Bob. Bob then applies the protocol unitary $V_y^{\prime(1)}$ to the message sent by Alice and sends the result back to her. Alice then applies her second unitary $W_x^{\prime(2)}$ and does a measurement. This protocol has communication $2k$ and accepts 0-inputs with probability at most $q$ and accepts 1-inputs with probability at least $p$, where $p>q$.

   Given any state $\ket{\phi_x}$, we can find an orthonormal basis $\beta_x=\{\ket{\beta_1}\cdots\ket{\beta_{2^k}}\}$ that includes $\ket{\phi_x}$ so that $\ket{\phi_x}$ is a member of the basis and $\sum_{i=1}^{2^k}\frac{\ket{\beta_i}\bra{\beta_i}}{2^k}=\frac{I}{2^k}$, such that the state $I/2^k$ is the uniform distribution on the elements in the basis. Consider a one-clean-qubit protocol that simulates the above quantum protocol and  goes as follows:
   \begin{enumerate}
       \item We define Alice's unitary $W_x^{\prime(1)}$ such that
       \begin{itemize}
           \item If $\ket{\beta_i}=\ket{\phi_x}$, then $W_x^{\prime(1)}:\ket{0}\ket{\beta_i} \mapsto\ket{1}\ket{\beta_i}$
           \item If $\ket{\beta_i}\neq\ket{\phi_x}$, then $W_x^{\prime(1)}:\ket{0}\ket{\beta_i} \mapsto\ket{0}\ket{\beta_i}$ \\and extend to a unitary in any possible way.
        \end{itemize}
   where $\ket{\beta_i}\in\beta_x$. Alice applies $W_x^{\prime(1)}$ to the initial state, in particular, computes $\sigma_x=W_x^{\prime(1)}(\ket{0}\bra{0}\otimes\frac{I}{2^c})W_x^{\prime(1)\dagger}$.
   \item Alice then sends the last $k$ qubits to Bob.
   \item Bob applies the unitary $V_y^{\prime(1)}$ to the qubits he received from Alice, in particular computes $\sigma_y=I\otimes V_y^{\prime(1)}(\sigma_x)I\otimes V_y^{\prime(1)\dagger}$, where dim($I$)=2. Bob sends the qubits back to Alice.
   \item Alice applies her unitary $W_x^{\prime(2)}$ (tensored with identity on the first qubit) to  $\sigma_y$ and measures the first qubit. She outputs 0 if she obtains a measurement result of $\ket{0}$. On the other hand, if she obtains a measurement of $\ket{1}$, she proceeds to execute the measurement of the original quantum protocol. In this case, the acceptance probability of 0-inputs is at most $\frac{q}{2^k}$ and the acceptance probability for 1-inputs is at least $\frac{p}{2^k}$. Note that the two measurements can be combined into one.
   \end{enumerate}
   The simulation of a $k$-clean-qubit quantum protocol by a one-clean-qubit protocol is shown in Figure \ref{fig:image2}:

   \begin{figure}[h]
   \begin{subfigure}{0.5\textwidth}
      \includegraphics[width=\linewidth, height=2cm]{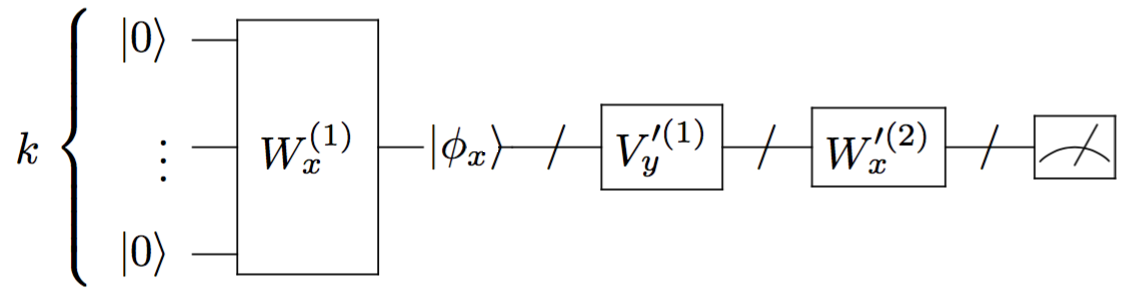}
      \caption{Original Quantum Protocol}
      \label{fig:subim1}
   \end{subfigure}
   \begin{subfigure}{0.5\textwidth}
      \includegraphics[width=\linewidth, height=2cm]{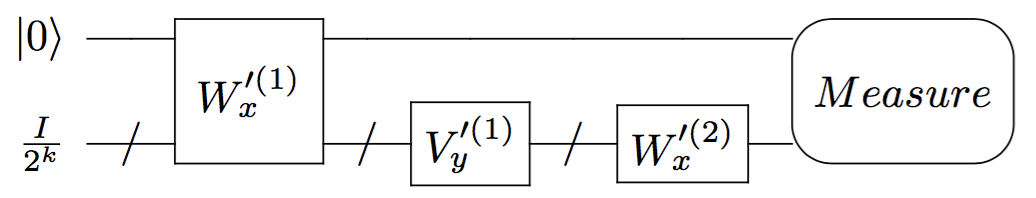}
      \caption{One-Clean-Qubit Protocol}
      \label{fig:subim2}
   \end{subfigure}
   \caption{Simulation by a one-clean-qubit protocol}
   \label{fig:image2}
   \end{figure}

\section{The Simulation Lower Bound}\label{app:low}

First, we insert dummies into the first $\frac{n}{2}-1$ entries of each string (set all to 1) and the remaining entries are drawn according to a distribution that will be defined in Fact \ref{dist}.

Consider the linear program (LP) for the rectangle bound (see \cite{Klauck}) as follows, where we set the acceptance probability for 1-inputs to be at least $\alpha=\frac{1}{n^3}$. We consider an additive error of $\frac{1}{n^4}$ and the simulation is required to accept 0-inputs with probability at most $\frac{1}{n^4}$ and accept 1-inputs with probability at least $\frac{2}{n^3}-\frac{1}{n^4}\geq\frac{1}{n^3}=\alpha$. Recall that we consider as 1-inputs only those $x,y$ with $\sum_i x_iy_i=\frac n2-1$, and as 0-inputs those with $\sum_i x_iy_i=\frac n2$. Denote by $\cal R$ the set of all rectangles in the communication matrix.
\\
   \underline{Primal}
    \begin{equation*}
      \begin{array}{ll@{}ll}
         \text{minimize}  & \displaystyle\sum\limits_{R\in \cal R} W_R &\\
         \text{subject to} &  \displaystyle\sum\limits_{\{R\in {\cal R}|x, y\in R\}} & W_R \geq\alpha,  & \mbox{for all } x,y: \sum_i x_iy_i=\frac n2-1\\
         & \displaystyle\sum\limits_{\{R\in {\cal R}|x, y\in R\}} & -W_R \geq-\frac{1}{n^4},  & \mbox{for all } x,y: \sum_i x_i y_i=\frac n2\\
         & &  W_R\geq 0
      \end{array}
   \end{equation*}
   \underline{Dual}
   \begin{equation*}
      \begin{array}{ll@{}ll}
         \text{maximize}  & \displaystyle\sum_{\{x, y|\sum_{i}x_i y_i=\frac{n}{2}-1\}} \alpha\gamma_{xy}-\sum_{\{x, y|\sum_{i}x_i y_i=\frac{n}{2}\}}\frac{1}{n^4}\sigma_{xy} &\\
         \text{subject to} & \displaystyle\sum_{\substack{\{x, y \in R|\sum_{i}x_i y_i=\frac{n}{2}-1\}}}\gamma_{xy}-\sum_{\substack{\{x, y\in R|\sum_{i}x_i y_i=\frac{n}{2}\}}}\sigma_{xy} & \leq 1 &
      \mbox{ for all } R\in \cal R \\
         & \hspace{80mm} \sigma_{xy}, \gamma_{x, y} & \geq 0\\
      \end{array}
   \end{equation*}
A protocol $\mathcal{P}$ that accepts 1-inputs with probability at least $\frac{1}{n^3}$ and accepts 0-inputs with probability at most $\frac{1}{n^4}$ can be viewed as a probability distribution on deterministic protocols. Each deterministic protocol (in a randomized public-coin protocol) can be represented by a protocol tree. The probabilities of decision trees are given as $p_1, p_2,\ldots,p_t$. Every leaf in each decision tree has an attached rectangle, and a decision: accept or reject. We consider only the rectangles which lead to acceptance, and we assign weight 0 to those rectangles that do not appear in any protocol tree at an accepting leaf and weight $W_R=\displaystyle\sum\limits_{\{i|R \mbox{\small\ accepted in  tree } i\}}p_i$ for rectangles appearing in protocol trees $i$.
   \begin{claim}
      The constraints in the primal LP hold.
   \end{claim}
   \begin{proof}
      \begin{itemize}
         \item Let $(x, y)$ be a 1-input. Summing up all the probabilities of the decision trees where $(x, y)$ is in a 1-rectangle, we get the LHS of the first inequality constraint, which also corresponds to the acceptance probability, which must exceed $\alpha$ on the RHS.
         \item Let $(x, y)$ be a 0-input. Adding up the probabilities of decision trees where $(x, y)$ appears in a 1-rectangle will give the LHS of the second inequality constraints, which is at most $1/n^4$ because that  is the maximum additive error allowed.
         \item The nonnegativity constraint is automatically fulfilled since $W_R$'s are sums of probabilities which must be at least zero.
      \end{itemize}
   \end{proof}
   \begin{claim}
      If there is a classical protocol that accepts 1-inputs with probability $\geq\alpha$ and 0-inputs with probability $\leq 1/n^4$ and communication $c$ then there exists a solution of cost $2^c$ for the primal LP.
   \end{claim}
   \begin{proof}
      The contribution of each decision tree to $W_R$ is at most $2^c\cdot p_i$, since there are at most $2^c$ leaves in each decision tree. Therefore,
      \[\sum_{R\in \cal R}W_R\leq\sum_{i=1}^{t}2^c\cdot p_i=2^c.\]
   \end{proof}

    Therefore, in a $\frac{1}{n^4}$-error simulation of a quantum protocol (that accepts 1-inputs with probability at least $\frac{2}{n^3}$ and accepts 0-inputs with probability 0), the simulating randomized protocol (with communication $c$) must accept 1-inputs with probability at least $\frac{2}{n^3}-\frac{1}{n^4}\geq \frac{1}{n^3}$ and accept 0-inputs with probability at most $\frac{1}{n^4}$, and hence yield a solution to the primal LP of cost at most $2^c$. By LP duality the primal and its dual have the same cost, and we want to show the lower bound for the cost. Hence, we work with the dual.

   In the dual, both $\gamma_{xy}$ and $\sigma_{xy}$ are nonzero if $x_1=\cdots=x_{\frac{n}{2}-1}=y_1=\cdots=y_{\frac{n}{2}-1}=1$ and are zero otherwise.
   \begin{fact}[Razborov's distribution on 0- and 1-inputs for Disjointness on $\frac{n}{2}+1$ inputs \cite{recbound}]\label{dist}\hfill
       \begin{itemize}
           \item $\mu_1$ (distribution on 1-inputs)\\
           $(x, y)$ is chosen uniformly at random subject to:
           \begin{itemize}
               \item $x$, $y$ each have exactly $\frac{n/2+1}{4}$ 1's
               \item There is no index $i\in\{1,2,\cdots ,\frac{n}{2}+1\}$ in which $x_i=y_i=1$.
           \end{itemize}
           \item $\mu_0$ (distribution on 0-inputs)\\
           $(x, y)$ is chosen uniformly at random subject to:
           \begin{itemize}
               \item $x$, $y$ each have exactly $\frac{n/2+1}{4}$ 1's
               \item There is exactly one index $i\in\{1,2,\cdots ,\frac{n}{2}+1\}$ in which $x_i=y_i=1$.
           \end{itemize}
       \end{itemize}
   \end{fact}
   From \cite{recbound}, the rectangle bound for our problem is as follows:
   \begin{center}
       For all rectangles $R=A\times B$ with $A, B\subseteq\{0, 1\}^{\frac{n}{2}+1}$, there exist constants $\epsilon,\delta>0$ such that
   \end{center}
   \begin{equation}\label{bound}
           \mu_0(R)\geq\epsilon\cdot\mu_1(R)-2^{2\delta n}.
       \end{equation}
    For a rectangle $R=A\times B$ where $A, B\subseteq\{0, 1\}^n$ let $\Tilde{R}\subseteq R$, where $\Tilde{R}$ is the subrectangle in which all $x_1=\cdots=x_{\frac{n}{2}-1}=y_1=\cdots=y_{\frac{n}{2}-1}=1$ and $\Tilde{x}=x_{\frac{n}{2}}\cdots x_n, \Tilde{y}=y_{\frac{n}{2}}\cdots y_n$ denote the substrings of  $x$ and $y$ which have length $\frac{n}{2}+1$ each.

    We seek   a  solution of the dual. The following are only true if $x_1=\cdots=x_{\frac{n}{2}-1}=y_1=\cdots=y_{\frac{n}{2}-1}=1$, otherwise $\gamma_{x,y},\sigma_{xy}=0$:

   \begin{itemize}
       \item 1-inputs: $\gamma_{xy}=\mu_1(\Tilde{x},\Tilde{y})\cdot 2^{\delta n}$
       \item 0-inputs: $\sigma_{xy}=\mu_0(\Tilde{x},\Tilde{y})\cdot 2^{\delta n}\cdot\frac{\alpha}{\frac{1}{n^4}}\cdot\frac{1}{10}=\mu_0(\Tilde{x},\Tilde{y})\cdot 2^{\delta n}\cdot\frac{n}{10}$
   \end{itemize}
   Now, we check if all the constraints are satisfied. we analyze the following cases:
   \begin{itemize}
       \item $\mu_1(\Tilde{R})\leq 2^{-\delta n}$:\\
       \begin{align*}
           \begin{split}
               & \displaystyle\sum_{\substack{ \{x,y  \in R   |\sum_{i}x_i  y_i=\frac{n}{2}-1\}}}\gamma_{xy}-\sum_{\substack{\{x,y \in R|\sum_{i}x_i  y_i=\frac{n}{2}\}}}\sigma_{xy} \\
            \leq & \displaystyle\sum_{\substack{ \{x,y \in R |\sum_{i}x_i y_i=\frac{n}{2}-1\}}}\gamma_{xy} \\
          \end{split}
      \end{align*}
      \begin{align*}
          \begin{split}
               & = \displaystyle\sum_{\substack{x, y\in R:\\|\Tilde{x}\wedge \Tilde{y}|=0\\x_1=\cdots=x_{\frac {n}{2}-1}=y_1=\cdots=y_{\frac{n}{2}-1}=1}}\gamma_{xy}\\
               & = \displaystyle\sum_{\substack{x, y\in R:\\|\Tilde{x}\wedge \Tilde{y}|=0\\x_1=\cdots=x_{\frac {n}{2}-1}=y_1=\cdots=y_{\frac{n}{2}-1}=1}}\mu_1(\Tilde{x},\Tilde{y})\cdot 2^{\delta n}\\
               & = \mu_1(\Tilde{R})\cdot 2^{\delta n}\\
               & \leq 1
           \end{split}
       \end{align*}
       \item $\mu_1(\Tilde{R})\geq 2^{-\delta n}$:\\
       \begin{align*}
           \begin{split}
               & \displaystyle\sum_{\substack{\{x, y \in R|\sum_{i}x_i y_i=\frac{n}{2}-1\}}}\gamma_{xy}-\sum_{\substack{\{x, y\in R|\sum_{i}x_i y_i=\frac{n}{2}\}}}\sigma_{xy}\\
               & = \displaystyle\sum_{\substack{x, y\in R:\\|\Tilde{x}\wedge \Tilde{y}|=0\\x_1=\cdots=x_{\frac{n}{2}-1}=y_1=\cdots=y_{\frac{n}{2}-1}=1}}\gamma_{xy}
               -\displaystyle\sum_{\substack{x, y\in R:\\|\Tilde{x}\wedge \Tilde{y}|=1\\x_1=\cdots=x_{\frac{n}{2}-1}=y_1=\cdots=y_{\frac{n}{2}-1}=1}}\sigma_{xy}\\
               & = \displaystyle\sum_{\substack{x, y\in R:\\|\Tilde{x}\wedge \Tilde{y}|=0\\x_1=\cdots=x_{\frac{n}{2}-1}=y_1=\cdots=y_{\frac{n}{2}-1}=1}}\mu_1(\Tilde{x},\Tilde{y})\cdot 2^{\delta n}\\
               & -\displaystyle\sum_{\substack{x, y\in R:\\|\Tilde{x}\wedge \Tilde{y}|=1\\x_1=\cdots=x_{\frac{n}{2}-1}=y_1=\cdots=y_{\frac{n}{2}-1}=1}}\mu_0(\Tilde{x},\Tilde{y})\cdot 2^{\delta n}\cdot\frac{n}{10}\\
               & = 2^{\delta n}\big(\mu_1(\Tilde{R})-\frac{n}{10}\cdot\mu_0(\Tilde{R})\big)\\
               & \leq 0\\
           \end{split}
       \end{align*}
       since $\mu_1(\Tilde{R})\geq 2^{-\delta n}$ and hence  $\mu_0(\Tilde{R})\geq\frac{\epsilon}2{}\cdot\mu_1(\Tilde{R})$ from (\ref{bound})\footnote{Provided that $\epsilon\geq\frac{20}{n}$.}.
   \end{itemize}
   The functional constraints are satisfied for both cases.

   Substituting the value of $\sigma_{xy}$ and $\gamma_{x,y}$ into the objective function, we get
   \begin{align*}
       \begin{split}
           & \displaystyle\sum_{\{x, y|\sum_{i}x_i y_i=\frac{n}{2}-1\}} \alpha\gamma_{xy}-\sum_{\{x, y|\sum_{i}x_i y_i=\frac{n}{2}\}}\frac{1}{n^4}\sigma_{xy}\\
           & = \displaystyle\sum_{\substack{\{x, y|\sum_{i}x_i y_i=\frac{n}{2}-1\\x_1=\cdots=x_{\frac{n}{2}-1}=y_1=\cdots=y_{\frac{n}{2}-1}=1\}}} \alpha\gamma_{xy}-\displaystyle\sum_{\substack{\{x, y|\sum_{i}x_i y_i=\frac{n}{2}\\x_1=\cdots=x_{\frac{n}{2}-1}=y_1=\cdots=y_{\frac{n}{2}-1}=1\}}}\frac{1}{n^4}\sigma_{xy}\\
       \end{split}
   \end{align*}
   \begin{align*}
       \begin{split}
           & = 2^{\delta n}\cdot\Bigg(\sum_{\substack{\{x, y|\sum_{i}x_i y_i=\frac{n}{2}-1\\x_1=\cdots=x_{\frac{n}{2}-1}=y_1=\cdots= y_{\frac{n}{2}-1}=1\}}}\frac{1}{n^3}\cdot\mu_1(\Tilde{x},\Tilde{y})-\displaystyle\sum_{\substack{\{x, y|\sum_{i}x_i y_i=\frac{n}{2}\\x_1=\cdots=x_{\frac{n}{2}-1}=y_1=\cdots= y_{\frac{n}{2}-1}=1\}}}\frac{1}{n^4}\cdot\mu_0(\Tilde{x},\Tilde{y})\cdot\frac{n}{10}\Bigg)\\
           & = 2^{\delta n}\big(\frac{1}{n^3}-\frac{1}{10n^3}\big)\\
           & = 2^{\delta n}\cdot\frac{9}{10n^3}\\
           & = 2^{\Omega(n)}\\
       \end{split}
   \end{align*}
   This implies that the communication needed is at least
   $$\log({2^{\Omega(n)}})=\Omega(n).$$

   This finishes the proof of Theorem \ref{long}.

\section{Proof of Theorem \ref{ABC}}

Since $ABC$ is either $I$ or $-I$, it follows that tr($ABC)$ is either $n$ or $-n$, and hence either all diagonal entries $A_i BC_i$ are 1 or -1. The protocol below is performed for an arbitrarily chosen $i$.
\begin{enumerate}
    \item Let $A_i$ represent Alice's row vectors from matrix $A$ and $C_i$ represent Charlie's column vectors from matrix $C$. Charlie and Bob share a set of $2^{O(k)}$ random unit vectors $W_j\in S^{n-1}$ as a  public coin, where $k$ is a parameter to be determined later and  $S^{n-1}=\{x\in\mathbb{R}^n:\sum_i x_i^2=1\}$. Among the $2^{O(k)}$ vectors shared with Bob, Charlie computes $W_{max}=argmax_{W_j\in T}\{\braket{W_j|C_i}\}$.
    \begin{lemma}
       Define $T=\{W_j:W_j\in S^{n-1}\}$ as the set of vectors randomly drawn from $S^{n-1}$ under the Haar measure (the unique rotationally-invariant probability measure on $S^{n-1}$ ) such that $|T|={32\sqrt{k}}e^{2k}$. If $v\in S^{n-1}$ is a fixed vector, then there exists a $W_j\in T$ that has an inner product with $v$ that is greater than $\sqrt{\frac{k}{n}}$ with high probability, for  all $1\leq k\leq \frac{n}{4}$.
    \end{lemma}
    \begin{proof}
       According to Lemma 1 in \cite{caps}, Pr($\braket{v, W_j}^2\geq\frac{k}{n})\geq \frac{e^{-k}}{16\sqrt{k}}$ for $W_j\in S^{n-1}$ uniformly at random. We have  Pr($\braket{v, W_j}\geq\sqrt{\frac{k}{n}})\geq \frac{e^{-k}}{32\sqrt{k}}$ due to the fact that $\braket{v, W_j}$ could be negative. By the definition of $T$, we have that $$Pr(\forall W_j\in T:\braket{v, W_j}\leq\sqrt{\frac{k}{n}})\leq \big(1-\frac{e^{-k}}{32\sqrt{k}}\big)^{{32\sqrt{k}}e^{2k}}=\Big[\Big(1-\frac{1}{32\sqrt{k}e^{k}}\Big)^{32\sqrt{k}e^{k}}\Big]^{e^k}\leq({\frac{1}{e})}^{e^k}.$$ In other words, the probability of all $W_j$'s in the sample having an inner product with $v$ that is less than $\sqrt{\frac{k}{n}}$, is extremely small. This implies that there exists a $W_j\in T$ such that $\braket{v, W_j}\geq\sqrt{\frac{k}{n}}$ with high probability.
    \end{proof}

    Recall that $W_{max}$ is the vector that maximizes the inner product with $C_i$, then
    $$W_{max}=\alpha\ket{C_i}+\sqrt{1-\alpha^2}\ket{\sigma_i},$$
    where $\sigma_i\bot C_i$ and $\alpha\geq\sqrt{\frac{k}{n}}$.
    \item Next, Charlie sends the name (i.e., its index in $T$) of $W_{max}$ to Bob. This requires $O(k)$ bits of communication. Bob then computes the following: $$B\ket{W_{max}}=\alpha B\ket{C_i}+\sqrt{1-\alpha^2}B\ket{\sigma_i}.$$
    \item Alice and Bob then jointly estimate the inner product between $B\ket{W_{max}}$ and $A_i$ by using the protocol proposed by Kremer, Nisan and Ron \cite{knr:rand1round}.
    \begin{fact}[Inner Product Estimation Protocol by Kremer, Nisan and Ron \cite{knr:rand1round}]
       The inner product estimation protocol approximates the inner product between two vectors from $S^{n-1}$ within $\epsilon$ additive error, which requires communication $O(\frac{1}{\epsilon^2})$.
    \end{fact}
    \begin{align*}
        \begin{split}
            \bra{A_i}B\ket{W_{max}}
            & =\alpha \bra{A_i}B\ket{C_i}+\sqrt{1-\alpha^2}\bra{A_i}B\ket{\sigma_i}\\
            & = \pm\alpha+\sqrt{1-\alpha^2}\bra{A_i}B\ket{\sigma_i},\\
        \end{split}
    \end{align*}
    where $\sqrt{1-\alpha^2}\bra{A_i}B\ket{\sigma_i}=0$ since $\sigma_i\bot C_i$ and $B^T A^T_i$ is either equal to $C_i$ or $-C_i$. That is to say,
    \[\bra{A_i}B\ket{W_{max}}
          \begin{cases}
            \geq \sqrt{\frac{k}{n}},\hspace{1mm} \text{for}\hspace{1mm} \text{1-inputs}\\
            \leq -\sqrt{\frac{k}{n}},\hspace{1mm} \text{for}\hspace{1mm}\text{0-inputs}.\\
         \end{cases}
        \]
    Setting $\epsilon$ to be smaller than $\sqrt{\frac{k}{n}}$, say $\frac{1}{100}\sqrt\frac{k}{n}$, to allow for sufficient separation between 0- and 1-inputs, Kremer, Nisan and Ron's protocol requires $O(\frac{n}{k})$ communication.
    \end{enumerate}
    In order to minimize the total amount of communication ($O(k)$ in Step 2 and $O(\frac{n}{k})$ in Step 3), we set $k=\sqrt{n}$. Therefore, the total amount of randomized communication required for $ABC$ is $O(\sqrt{n})$.

   \end{appendix}

\end{document}